\documentclass[%
 reprint,
nofootinbib,
nobibnotes,
 amsmath,amssymb,
 aps,
pra,
onecolumn
]{revtex4-1}

\usepackage{graphicx}
\usepackage{dcolumn}
\usepackage{bm}
\usepackage{amssymb,amsfonts,latexsym,stmaryrd,amsthm}
\usepackage{epstopdf,refcount,nicefrac}
\usepackage[table]{xcolor}

\usepackage{proof}
\usepackage{epsfig}
\usepackage{euscript}
\usepackage{tikz,ifdraft}
\usetikzlibrary{arrows,positioning,matrix,backgrounds,calc,fit}
\usepackage{url, appendix, braket,array}

\newcolumntype{P}[1]{>{\centering\arraybackslash}m{#1}}

\newtheorem{thm}{Theorem}
\newtheorem{defn}{Definition}
\newtheorem{result}{Result}
\newtheorem{lemma}{Lemma}

\newcommand{\comment}[1]{}

\newcommand{\cc}{\mathcal{C}}
\newcommand{\cg}{\mathcal{G}}

\newcommand{\co}{\mathcal{O}}
\newcommand{\ch}{\mathcal{H}}
\newcommand{\cm}{\mathcal{M}}
\newcommand{\cb}{\mathcal{B}}

\newcommand{\empmod}{\mathcal{E}}

\definecolor{ngreen}{rgb}{0.2,0.7,0.2}
\definecolor{nred}{rgb}{0.9,0.1,0}
\definecolor{nblue}{rgb}{0.1,0.2,0.8}

\newcommand{\clg}{\cellcolor{lightgray}}

\begin{document}

\title{Graph-theoretic strengths of contextuality}
\author{Nadish de Silva}
\affiliation{Department of Computer Science, University College London,\\  WC1E 6BT London, United Kingdom}


\begin{abstract}
Cabello-Severini-Winter and Abramsky-Hardy (building on the framework of Abramsky-Brandenburger) both provide classes of Bell and contextuality inequalities for very general experimental scenarios using vastly different mathematical techniques. We review both approaches, carefully detail the links between them, and give simple, graph-theoretic methods for finding inequality-free proofs of nonlocality and contextuality and for finding states exhibiting strong nonlocality and/or contextuality.  Finally, we apply these methods to concrete examples in stabilizer quantum mechanics relevant to understanding contextuality as a resource in quantum computation.

\end{abstract}


\maketitle

\section{Introduction}

A fundamental question of quantum computation (QC) is to precisely identify the nonclassical features of quantum mechanics accounting for quantum advantages in computation.  This critical problem is often phrased in terms of finding properties characterizing those resource states capable of promoting a computational model to greater power, e.g. universal QC.  Understanding of quantum resources will yield both hardware efficiency gains and a theoretical basis for developing novel applications of quantum information.  Entanglement \cite{vidal}, superposition \cite{deutsch}, and discord \cite{datta} have been proposed as candidates but found unsatisfactory as explanations of quantum advantage.  Numerous striking results \citep{anders, galvao2005discrete, veitch,  delfosse2015wigner, raussrebit} have recently established \emph{contextuality} as necessary for magic state distillation (MSD) \cite{howard} and measurement-based quantum computation (MBQC) \cite{raussendorf}.

Our aim in this paper is to carefully delineate the connections between the vastly differing frameworks for contextuality used in these recent results in order to better understand the role of contextuality in QC and to conceptually clarify contextuality.  We use this connection to give directly computable graph-theoretic characterisations of the logical strengths of contextuality and apply these tools to examples in stabilizer quantum mechanics relevant to QC.

\subsection{Nonlocality \& contextuality}

Contextuality is a generalization of \emph{nonlocality}: the notion that the predictions of quantum mechanics do not admit a model that is classical in the sense of being locally causal.  That is, one in which, upon conditioning on a causal past, the joint distributions describing experiments performed at different sites factorize into distributions associated with each site \cite{bell1964einstein, wiseman2014two}.   Data from a two-site experiment are correlations $p(a,b|A,B)$ where $A, B$ and $a,b$ are measurement settings and outcomes respectively.  A locally causal model for such data is a space $\Lambda$ of hidden variables, a distribution $q$ on $\Lambda$, and conditional distributions $r_A(-|\lambda), r_B(-|\lambda)$ on outcomes that account for the data: $$p(a,b|A,B) = \sum_{\lambda \in \Lambda} \, q(\lambda) \cdot r_A(a|\lambda)r_B(b|\lambda).$$   A \emph{Bell inequality} is a bound on a weighted sum of correlations that is satisfied by all data a locally causal model.  Data not admitting any such model exhibits \emph{nonlocality}.  Bell gave the first such inequality and an entangled quantum state that violates it.  Decades later, entanglement and nonlocality became the basis of quantum communication protocols, e.g. superdense coding, quantum teleportation, and quantum cryptography \cite{benn1, benn2, benn3}. 

There are degrees of strength of nonlocality.  The correlations arising from a Bell state violate a Bell inequality.  Hardy states \cite{hardy1992quantum, hardy1993} preclude a locally causal model via logical arguments without the need for inequalities.  The Greenberger-Horne-Zeilinger (GHZ) state \cite{ghz} is, in a sense defined below, maximally nonlocal.  Nonlocality has been studied as a resource for communication tasks and it has been recognized that stronger nonlocality yields greater advantages \cite{vandamthesis,barrett, barrett2005nonlocal, buhrman2010nonlocality}. 

Contextuality
is the notion, due to Kochen-Specker \cite{ks}, that a system's observable properties cannot all be assigned deterministic outcomes in a manner independent of the method of observation used to acertain each property.  It subsumes nonlocality as a special case since Fine's theorem \cite{fine} tells us that data admitting a locally causal model also admit one in which the distributions $r_A(-|\lambda), r_B(-|\lambda)$ are deterministic.  Like nonlocality, contextuality admits a description via  inequalities, e.g. Klyachko \emph{et al.}'s \emph{contextuality inequality} \cite{klyachko}.  Unlike nonlocality, contextuality can manifest in single-site systems.

\subsection{Contextuality as a resource in QC}

Recently, the hypothesis that contextuality plays a critical role in enabling quantum computation---analogous to the role nonlocality  plays in superclassical communication---has received considerable attention.  For example, Howard \emph{et al.} \cite{howard} showed that contextuality is a necessary criterion in finding those quantum states that are suitable for the magic state distillation (MSD) protocol \cite{bravkit} which promotes a scheme of classical computing power to (fault-tolerant) quantum universality.   

A natural question is: does stronger contextuality yield greater computational advantages?  Positive evidence for this question is given by Raussendorf \cite{raussendorf} in the setting of measurement-based quantum computers (MBQCs) \cite{raussendorf2001one}.  Further progress on the resource theory of contextuality necessitates clarification of the strengths of contextuality.

The aforementioned results on contextuality as a resource rely on mathematical tools developed by Cabello-Severini-Winter (CSW) \cite{CSW} and Abramsky-Brandenburger (AB) \cite{AB} which promote the example-based understanding of nonlocality and contextuality to a higher-level, structural one.  They apply in any experimental scenario and both provide a complete set of Bell/contextuality inequalities.  However, they differ vastly in their approach and how they relate to one another is not immediately clear.  The plurality of perspectives and vernaculars for describing contextuality constitutes an obstacle to progress on a computational resource theory of contextuality.

\subsection{Overview}

We first review some definitions and results of CSW and AB before we

\begin{enumerate}
\item Precisely detail the links between their two furnished classes of Bell/contextuality inequalities.
\item Cast the powerful, highly abstract, logical hierarchy of strengths of contextuality due to AB in terms of graph invariants in order to render them more tractable as theoretical tools and enable direct computation using standard software libraries.
\item Pose the question, does the logical strength of contextuality of a state correlate with its usefulness as a resource?
\end{enumerate}

We provide three computational examples to demonstrate how these tools may be used to gain insight into the contextuality in QC program.

The first result shows that the sort of contextuality witnessed by logical contradictions (rather than mere violation of contextuality inequalities) may not be useful in identifying the single-qutrit states suitable for magic state distillation.  However, we provide evidence that it may be useful in identifying higher-qudit magic states.
The third result addresses the question that Howard \emph{et al.} end their paper with:
\begin{quote}In the qubit case, it is a pressing open question whether a suitable operationally motivated refinement or quantification of contextuality can align more precisely with the potential to provide a quantum speed-up. \end{quote}
Our result asserts that the quantification sought would necessarily need to discriminate between states that are all maximally contextual.  One might hope to achieve this by considering measures specialized to the particular setting of stabilizer QM, e.g. \cite{howard2016application}.

As we shall see, nonlocality is the special case of contextuality where the observables are compatible (i.e. comeasurable) precisely when they correspond to spatially separated local measurements.  Thus, in the remainder of this paper, we use the terminology contextuality and contextual inequalities to include the special cases of nonlocality and Bell inequalities.


\section{General approaches to contextuality}

Both the AB and CSW approaches begin by abstractly describing physical experiments without making an assumption of simultaneous comeasurability of all properties and by describing how operational data from such an experiment is tabulated.  The primary question asked of operational data is whether it can, in principle, be reproduced by a model in which observables simultaneously possess deterministic values.  

AB's and CSW's representations of experiments and data differ significantly.  Before delineating their connections, we review the key structures and results of each.


\subsection{CSW inequalities}CSW employ graph theory to study nonlocality and contextuality for general experimental scenarios.  (The basic graph-theoretic definitions necessary for understanding what follows are found in Appendix A).  An experiment is formalized by an \emph{exclusivity graph}.  The vertices of such a graph represent events, or, alternatively, answers to propositions about a system answerable by the experiment.  Vertices that are adjacent represent \emph{mutually exclusive} events.

Operational data coming from repeatedly performing the experiment on a fixed preparation of the system is tabulated using probabilities $p_i$ for each vertex $v_i$.  CSW inequalities are upper bounds on a linear combination of these probabilities: $\Sigma_i w_i p_i$ where $p_i$ is the likelihood of observing the $i$-th event and $w_i \geq 0$ is a coefficient.  Such a linear combination is compactly represented as a weighted graph $(G,w)$.

CSW define a \emph{classical model} for operational data represented by a graph $G$ and probabilities $p_i$ as a classical sample space $\Lambda$, an event $e_i \subset \Lambda$ for each vertex $v_i$ such that the events corresponding to adjacent vertices are mutually exclusive (i.e. disjoint subsets of $\Lambda$), and a probability distribution $\mu$ on $\Lambda$ such that $p_i$ is the probability that $\mu$ assigns to $e_i$.  

For a linear combination $\Sigma_i w_i p_i$, represented by the weighted graph $(G,w)$, one can ask the question: what is the maximum this sum can achieve if the data $p_i$ has a classical model?  CSW show that this maximum is precisely the weighted independence number of $(G,w)$.  Data that violate any one of these bounds, for some choice of coefficients $w_i$, cannot be explained with a classical model.  Therefore, CSW contextuality inequalities are of the form $$\sum_i w_i p_i \leq \alpha(G,w).$$

A \emph{quantum model} is a Hilbert space $\ch$, projectors $P_i \in \cb(\ch)$ for each vertex $v_i$ such that adjacent vertices are represented by orthogonal projectors, and a pure state $\ket{\psi} \in \ch$ such that the probabilities $p_i$ arise as $\bra{\psi} P_i \ket{\psi}$.  In this case, the maximum  $\Sigma_i w_i p_i$ can achieve over all quantum models (i.e. the Cirel'son-type bound) is bounded above by the Lov\'asz theta number: $\Sigma_i w_i p_i \leq \vartheta(G,w)$.

Finally, a \emph{generalized model}\footnote{These models are those that satisfy the \emph{consistent exclusivity} or \emph{E1} principle.} is simply data $p_i$ such that the sum of probabilities corresponding to all the vertices in a clique is less than or equal to 1.  The maximum  $\Sigma_i w_i p_i$ can achieve over all generalized models is given by the fractional packing number: $\Sigma_i w_i p_i \leq \alpha^*(G,w)$.


\subsection{The AB sheaf-theoretic approach}
In the AB approach\footnote{We give an elementary introduction with simplified terminology.  Curious readers are encouraged to read \cite{AB}.}, an experiment is formally described by a \emph{measurement scenario}: a pair $(\cm, \cc)$ where $\cm$ is a set of abstract labels for measurements and $\cc$ is the set of \emph{contexts}.  A context is a set $C \subset \cm$ representing a maximal set of compatible (i.e. comeasurable) measurements; thus, it is required that if $C$ is a context, no proper subset of $C$ is also a context.  It is further required that every measurement $m \in \cm$ is contained in at least one context.  Measurement of each individual $m \in \cm$ yields a value from the outcome set $\co$ (usually $\{0,1\}$).

For example, the standard Bell scenario is described with $\cm = \{A_0, A_1, B_0, B_1\}$ and contexts $\cc = \{\{A_0,B_0\}$, $\{A_0,B_1\}$, $\{A_1,B_0\}$, $\{A_1,B_1\}\}$.

A \emph{formal event} is a function $e: S \to \co$ from a set $S \subset \cm$ of measurements to outcomes; in other words, a joint outcome for all of the measurements in $S$.  Note that this is a strictly mathematical construction as it is not assumed that $S$ is a subset of a context.  In other words, the measurements in $S$ are not assumed to be physically comeasurable.  The set of all formal events for a fixed $S$ is denoted by $E(S) = \{e: S \to \co\}$.  Whenever $S' \subset S$, a \emph{coarse-graining} of a formal event $e: S \to \co$ is defined by forgetting the outcomes for measurements in $S$ but not in $S'$; it is denoted by $e|_{S'}$ and explicitly defined by $e|_{S'}(M') = e(M')$ for  $M' \in S'$.  

The partial distribution sets $D(S)$ are then defined as the set of all probability distributions on the formal event set of $S$, i.e. $D(S) = \{ p: E(S) \to [0,1] \,|\, \Sigma_{e \in E(S)} p(e) = 1\}$.  Whenever $S' \subset S$, the \emph{marginal  distribution} of $p \in D(S)$ is defined by averaging over the the information about measurements in $S$ but not in $S'$; it is denoted by $\mu^S_{S'}(p)$ and explicitly defined, for  $e' \in E(S')$, by $$\mu^S_{S'}(p)(e') = \sum_{e|_{S'} = e'}^{e \in E(S)}  p(e).$$  

Operational data from repeated experiments on a system in a fixed preparation are tabulated by families of probability distributions $\empmod_C \in D(C)$ that are indexed by contexts $C \in \cc$.  They describe the likelihoods of joint outcomes for all the maximal sets of comeasurable measurements.  Such data are \emph{nonsignalling} (also known as \emph{nondisturbing}) when they yield common marginal distributions for each intersection of contexts: $\mu^C_{C \cap C'}(\empmod_{C}) = \mu^{C'}_{C \cap C'}(\empmod_{C'})$ for any pair $C,C'$ of contexts.  When data obeys this nonsignalling condition, the distributions $\empmod_C$ constitute an \emph{empirical model} $\empmod$.

\begin{table}[h!]
 \begin{tabular}{|| P{0.4cm}P{0.4cm} | P{0.4cm} P{0.4cm} P{0.4cm} P{0.4cm} ||} 
 \hline
 $A$ & $B$ & 00 & 01 & 10 & 11\\ 
 \hline
 $A_0$ & $B_0$ & \clg \nicefrac{1}{2} & 0 & 0 & \clg \nicefrac{1}{2} \\ 
 $A_0$ & $B_1$ & \clg\nicefrac{3}{8} & \nicefrac{1}{8} & \nicefrac{1}{8} & \clg\nicefrac{3}{8} \\ 
 $A_1$ & $B_0$ & \clg\nicefrac{3}{8} &  \nicefrac{1}{8} & \nicefrac{1}{8} & \clg\nicefrac{3}{8} \\ 
 $A_1$ & $B_1$ & \nicefrac{1}{8} & \clg \nicefrac{3}{8} & \clg \nicefrac{3}{8} & \nicefrac{1}{8} \\ 
  \hline
\end{tabular}
\quad \quad
 \begin{tabular}{|| P{0.4cm}P{0.4cm} | P{0.4cm} P{0.4cm} P{0.4cm} P{0.4cm} ||} 
 \hline
 $A$ & $B$ & 00 & 01 & 10 & 11\\ 
 \hline
 $A_0$ & $B_0$ & \clg\nicefrac{1}{2} & 0 & 0 & \clg\nicefrac{1}{2} \\ 

 $A_0$ & $B_1$ & \clg\nicefrac{1}{2} & 0 & 0 & \clg\nicefrac{1}{2} \\ 
 $A_1$ & $B_0$ & \clg\nicefrac{1}{2} & 0 & 0 & \clg\nicefrac{1}{2} \\ 
 $A_1$ & $B_1$ & 0 & \clg\nicefrac{1}{2} & \clg\nicefrac{1}{2} & 0 \\ 
  \hline
\end{tabular}
\caption{Bell state correlations (L); PR box correlations (R).  Each row is a context and a distribution on joint outcomes.}
\end{table}

An empirical model is local/noncontextual precisely when its predictions can be accounted for by a locally causal/noncontextual hidden variable model.  That is, there is a hidden variable space $\Lambda$, a distribution $q$ on $\Lambda$, and conditional distributions $r_M(-|\lambda)$ on $\co$ for each $M \in \cm$ such that for any context $C = \{M_1, ..., M_N\}$ and $e: C \to \co$: $$ \empmod_C(e) = \sum_{\lambda \in \Lambda} \, q(\lambda) \cdot r_{M_1}(o_1|\lambda) \cdots r_{M_N}(o_N|\lambda),$$ where $o_i = e(M_i)$ is the outcome of the $i$-th measurement.  Equivalently, an empirical model $\empmod$ is noncontextual when there is a joint distribution $J \in D(\cm)$ yielding the $\empmod_C$ as marginals: $\empmod_C = \mu^\cm_C(J)$.  Thus, the hidden variable space can be taken to have the canonical form of $\Lambda = E(\cm)$ whose \emph{canonical hidden variables} are functions $\lambda: \cm \to \co$ and the $r_M(-|\lambda)$ are the deterministic distributions assigning probability 1 to the outcome $\lambda(M)$ and 0 to all others.  That an empirical model that arises as a family of marginals distributions of a single joint probability distribution on all measurements always admits such a canonical model is a vast generalization of Fine's theorem to all experimental scenarios.  

Abramsky-Hardy \cite{logicalbi} give a complete set of Bell and contextuality inequalities via logical consistency; AB describe two natural logical strengths of contextuality.  We detail these tools, and how they may be understood in terms of graph invariants, below.

\section{The exclusivity graph of a measurement scenario}
In this section, we build a CSW exclusivity graph from AB's description of an experiment.  Explicitly corresponding the two representations of experiments and data facilitates the synthesis of insights and tools of each approach.

An \emph{observable event} is an $e \in E(C)$ for some context $C \in \cc$.  They are the elementary events observable in an experiment, or, alternatively, answers to the most refined questions one can ask of a system.  Given an empirical model $\empmod$, an observable event $e$ is \emph{impossible} when $\empmod_C(e) = 0$; otherwise, it is \emph{possible}.

Given a probability table for an empirical model, the observable events correspond to the (empty) cells.  An empirical model provides the actual probabilities. Two observable events are called \emph{mutually exclusive} or \emph{inconsistent} when they assign different outcomes to a shared measurement.

\begin{defn}The \emph{exclusivity graph $\cg(\cm, \cc)$ of the measurement scenario} $(\cm, \cc)$ has vertices given by the observable events $e: C \to \co$ where $C \in \cc$ is a context.   Two vertices $e_1: C_1 \to \co$ and $e_2: C_2 \to \co$ are adjacent whenever $e_1$ and $e_2$ are \emph{inconsistent}; that is, whenever $e_1(M) \neq e_2(M)$ for some measurement $M$ in both $C_1$ and $C_2$.
\end{defn}

Ac\'in \emph{et al.} associate a hypergraph to a measurement scenario via \emph{measurement protocols} \cite[D.1.4]{acin}.  The nonsignalling conditions of empirical models are encoded as clique constraints on this hypergraph's associated non-orthogonality graph \cite[2.3.1]{acin}.  Definition 1 is isomorphic to the complement of the graph built this way; see Appendix C for a proof.  Here, we focus on directly constructing the exclusivity graph of a measurement scenario in order to understand sheaf-theoretic tools in terms of graph invariants.

\begin{figure}[!htb]
\centering
\includegraphics[width=0.41\textwidth]{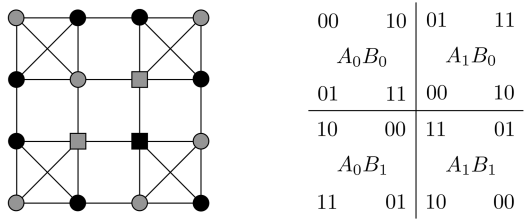}
\caption{The exclusivity graph for the standard Bell scenario (L); labels of measurement settings and outcomes (R).  Straight lines connect exclusive events.  The standard CHSH inequality \cite{chsh} is derived by giving the grey vertices the weight 1 and the others 0; the independence number is 3.  Circular vertices are the possible events of a Hardy model.  Grey vertices are the possible events of a PR box \cite{popescu1994quantum}.}
\end{figure}


\begin{defn}The \emph{support graph} of an empirical model $\empmod$ is the induced subgraph $\cg_\empmod$ of $\cg(\cm, \cc)$ given by retaining only the possible events.
\end{defn}

These definitions generalize constructions used by Sadiq \emph{et al.} \cite{sadiq} and Fritz \emph{et al.} \cite{fritz2013local, sainz2014exploring} for analysing nonlocality to the setting of contextuality.  

In our proofs, we repeatedly exploit the following key connection between the approaches of AB and CSW:

\begin{lemma}
For any measurement scenario $(\cm,\cc)$, there is a bijective correspondence between canonical hidden variables $\lambda: \cm \to \co$ and independent sets $I \subset \cg(\cm,\cc)$ of size $|\cc|$. 
\end{lemma} 
As the observable events of a common context form a clique, the size of an independent set in $\cg(\cm,\cc)$ can contain at most one observable event from each clique.  Thus, the size of an independent set is bounded by the number of contexts.  Independent sets whose size achieves this upper bound must contain precisely one event of each context.

Given such an independent set $I$, one can construct a canonical hidden variable $\lambda_I: \cm \to \co$ by defining $\lambda_I(M) = e(M)$ where $e: C \to \co$ is any event in $I$ such that $M \in C$.  If $e': C' \to \co$ is another event in $I$ such that $M \in C'$, the fact that $e$ and $e'$ are not adjacent tells us that $e(M) = e'(M)$; thus, $\lambda_I$ is well-defined.  Conversely, given a canonical hidden variable $\lambda: \cm \to \co$, the set $I_\lambda = \{\lambda|_C: C \to \co \;|\; C \in \cc\}$ of its coarse-grainings form an independent set.

\section{Logical Bell inequalities}The \textit{logical Bell inequalities} of $(\cm, \cc)$ are a distinguished class of Bell or contextuality inequalities identified by Abramsky and Hardy \cite{logicalbi}.  
Working in the setting of measurements with binary outcomes\footnote{The case of measurements with different outcome sets is easily reduced to the binary case.} ($\co = \{0,1\}$), we can encode observable events $e$ as propositional formulae $f(e)$ by viewing the measurements $M \in \cm$ as Boolean variables\footnote{The symbols $\wedge, \vee, \neg$ represent AND, OR, NOT respectively.}: $$f(e) = \bigwedge\limits_{M \in C} \,
  \begin{cases} 
      \hfill M    \hfill & \text{ if $e(M) = 1$} \\
      \hfill \neg M \hfill & \text{ if $e(M) = 0$} \\
  \end{cases}\;. $$  For example, the two grey events of the last row in Table 1 are individually encoded as $A_1 \wedge \neg B_1$ and $\neg A_1 \wedge B_1$.  
 
Two vertices are adjacent precisely when their formulae are logically inconsistent, i.e. one cannot assign \emph{true} or \emph{false} to the variables in $\cm$ in a way making both their formulae true.

 Given $N$ Boolean formulae $f_i$ that are true with probability $p_i$, if the formulae are logically inconsistent, then it follows from elementary probability theory \citep[\S I.A]{logicalbi} that $$\sum p_i \leq N - 1.$$  A logical Bell inequality is built by choosing a subset of the observable events $E_i \subset E(C_i)$ from each context $C_i \in \cc$ and constructing the formulae
$$ f_i = \bigvee\limits_{e \in E_i} \, f(e) \;. $$  When the formulae $f_i$ contain a contradiction, we obtain an inequality: the sum of the likelihoods of all of the events in all of the $E_i$ is less than or equal to $N - 1$.

For example, the grey subset of events of the last row (in either table) of Table 1 is encoded as $f_4 = (A_1 \wedge \neg B_1) \vee (\neg A_1 \wedge B_1)$ or, equivalently, $A_1 \iff \neg B_1$.  The CHSH inequality is derived by noting that the formulae for all four rows ($A_0 \iff B_0$, $A_0 \iff B_1$, $A_1 \iff B_0$, $A_1 \iff \neg B_1$) are contradictory and so the sum of probabilities of all the grey events has a classical bound of 3.  The Bell state violates this bound by \nicefrac{1}{4}.  The PR box maximally violates it by 1.

\begin{thm}
Logical Bell inequalities can be seen as examples of CSW inequalities on the exclusivity graph $\cg(\cm, \cc)$.\footnotemark[\getrefnumber{proofB}]
\end{thm}

The inequalities are constructed by giving the events of $E_i$ the  integer weight 1 and all other events weight 0.  The CSW bound given by the independence number is tight whereas the AB bound of $N-1$ need not be.  

Abramsky-Hardy also define \emph{extended logical Bell inequalities} by choosing sets of observable events $E_i \subset E(C_i)$ whose formulae $f_i$ are contradictory as above; however, the coefficients in these inequalities are allowed to be any integer $k_i \geq 0$.  The upper bound they give, defined in terms of logical consistency, is, in fact, simply the CSW bound given by the weighted independence number.  As these extended inequalities have the same integer coefficient for all observable events of a common context, they form a much smaller class of contextual inequalities than all possible CSW inequalities on $\cg(\cm,\cc)$ which allow each individual observable event its own coefficient.  However, they still form a \emph{complete} class of contextuality inequalities in the sense that knowing that an empirical model satisfies only the extended logical Bell inequalities is sufficient to conclude that it is noncontextual.

The logical interpretation of contextuality inequalities given by Abramsky-Hardy can be extended to all CSW inequalities: the weighted independence number can be seen as maximizing the weighted sum over logically consistent formulae.

\section{Hardy's paradox \& logical contextuality}Considerable work has been done on establishing Bell's theorem without inequalities \cite{ghz, hardy1993,cabello2001bell,cabello2002bell,yokota2009direct,aharonov2002revisiting}.  The idea is to preclude locally causal models of quantum theory using knowledge of the mere possibility or impossibility of certain correlated events rather than the precise correlations.  While the first proof of Bell's theorem without inequalities was the GHZ state, Hardy's paradox \cite{hardy1992quantum} better exemplifies the subtle logical contradictions at play.  As inequality-free proofs of nonlocality cannot be accomplished with standard Bell states, Hardy's paradox can be interpreted as a stronger form of nonlocality than the mere statistical one.  In the AB framework, this notion is generalized to \emph{logical nonlocality and contextuality}; we exploit this to give a graph-theoretic characterization of operational data that preclude locally causal/noncontextual models without resort to inequalities.

A state is logically nonlocal/contextual if any locally causal/noncontextual model that predicts the occurrence of all empirically observed events must also predict the occurrence of an event that is not empirically observed.  More formally: there is a possible event $e$ such that any canonical hidden variable $\lambda: \cm \to \co$ that is compatible with $e$ (i.e. the coarse-graining $\lambda|_C$ is $e$) must also predict the occurrence of an impossible event (there is a context $C' \in \cc$ for which $\lambda|_{C'}$ is impossible).  A state that satisfies all logical Bell inequalities cannot be logically contextual. 

Finding a graph-theoretic characterisation of logical contextuality requires defining a new graph invariant heretofore unused in the CSW approach: 
\begin{defn}A The \emph{independence degree} of a vertex $v$ in a finite graph $G$ is the size of the largest independent set in $G$ that contains $v$; the \emph{minimal independence number} is the minimum independence degree over all vertices.
\end{defn}

Note that the standard independence number of a graph is the maximum independence degree over all vertices.

\begin{thm}
\label{hardy}
An empirical model $\empmod$ is logically nonlocal or contextual (i.e. admits an inequality-free proof of nonlocality or contextuality) if and only if the minimal independence number of its support graph $\cg_\empmod$ is less than the number of contexts.\footnote{\label{proofB}The proof is found in Appendix B}
\end{thm}

In Figure 1, that there is no independent set of 4 circular vertices that includes the top-left corner vertex proves that Hardy's model is logically nonlocal

\section{Maximal nonlocality \& strong contextuality}The notion of \emph{maximal nonlocality} was introduced by Elitzur-Popescu-Rohrlich \cite{eprohr} for states in a standard Bell experiment and extended to more general nonlocality experiments by Barrett \emph{et al.} \cite{barrettmax}.  AB generalized this notion to maximal contextuality and gave an equivalent logical notion called \emph{strong contextuality} which we exploit to give a simple graph-theoretic criterion.

Well-known examples of maximally nonlocal empirical models include GHZ states and PR boxes.  Access to maximally nonlocal resources yields advantages for communication tasks.  For example, sharing an unlimited number of PR boxes between two parties renders all communication complexity problems trivial \cite{vandamthesis}. 

Strong contextuality is known to play a role in quantum computation:  Raussendorf \cite{raussendorf} showed that strong contextuality is necessary for allowing certain MBQCs to deterministically compute nonlinear functions.  

An  empirical model $\empmod$ can be convexly decomposed into a noncontextual part $A$ and a nonsignalling part $Z$:
$$\empmod = \tau A + (1-\tau)Z.$$
An empirical model is \textit{maximally contextual} if admits no decomposition with nonzero $\tau$.  

This is equivalent to strong contextuality: there is no canonical hidden variable $\lambda: \cm \to \co$ such that the observable event $\lambda|_C$ is possible for all contexts $C \in \cc$.  An empirical model is strongly contextual precisely when it maximally violates (by 1) the logical Bell inequality yielded by choosing $E_i$ to be its possible events in $E(C_i)$.  Strongly contextual empirical models are precisely the convex sums of the nonsignalling polytope's contextual vertices \cite{acin, posspoly}.

\begin{thm}
An empirical model $\empmod$ is strongly nonlocal or contextual if and only if the independence number of its support graph $\cg_\empmod$ is less than the number of contexts.\footnotemark[\getrefnumber{proofB}]
\end{thm}

In Figure 1, that there is no independent set of 4 grey vertices  proves that the PR box is strongly nonlocal. 

\section{Applications}
We provide computational examples wherein we apply Theorems 2 and 3 to study the contextuality of quantum states relative to all stabilizer operations \cite{gotthesis}.

Recent work of Howard \emph{et al.} \cite{howard} has shown contextuality with respect to two-qudit stabilizer operations to be a necessary condition for MSD in odd-prime-power qudit systems.  Specifically, they expressed the inequalities describing the faces of the polytope of quantum states with positive Wigner function (in the Gross representation \cite{gross2006hudson}) as CSW inequalities on the orthogonality graph of two-qudit stabilizer projectors.  As a classically efficient simulation algorithm is known for stabilizer quantum theory augmented with non-stabilizer states having positive Wigner representation \cite{veitch}, any quantum state that is noncontextual with respect to stabilizer projectors cannot serve as a resource to promote stabilizer operations to super-classical computational power.

\subsection{Qutrit stabilizer operations}

In the odd-prime-power qudit cases, where contextuality is a useful necessary criterion for identifying MSD resource states, we ask whether stronger forms of contextuality are also useful.  Following Howard \emph{et al.}, we consider contextuality with respect to two-qutrit stabilizer operations.

\begin{result}
The states $\ket{M} \otimes \ket{M}$, where $\ket{M}$ is the qutrit magic state $\ket{M_0}$, $\ket{E}$, $\ket{N'}$, or $\ket{\psi_{U_v}}$ \cite{anwar,hillary, jiri} are neither logically contextual nor strongly contextual with respect to two-qutrit stabilizer operations.
\end{result}

Therefore, logical contextuality of $\ket{M} \otimes \ket{M}$ is not a necessary condition for a qutrit state $\ket{M}$ to serve as a resource promoting stabilizer operations to universal QC.

\begin{result}
The two-qutrit magic state $$\ket{CS}= \sum_{j,k = 0}^2 e^{\frac{2 \Pi i}{3} j k^2} \ket{j} \otimes \ket{k}$$ \cite{markpc} is strongly contextual with respect to two-qutrit stabilizer operations.
\end{result}

This suggests that stronger forms of contextuality may play a role in identifying higher-qudit magic states; we leave further investigation of this to future work.  As $\ket{CS}$ deterministically yields a non-Clifford gate, this result strengthens the connection between strong contextuality and advantage in deterministic computation and communication tasks that is suggested by the essential role strong contextuality plays in MBQC \cite{anders,raussendorf} and zero-error classical communication \cite{cubitt}.  It seems reasonable to conjecture that any magic state exhibiting such a deterministic property must be strongly contextual.

The support graphs of the aforementioned states are induced subgraphs of the orthogonality graph of two-qutrit stabilizer projectors which was was computed using Gross' discrete phase space \cite{gross2006hudson}.  The phase space for a two-qutrit system is given by $V = (\mathbb{Z}_3 \times \mathbb{Z}_3) \times (\mathbb{Z}_3 \times \mathbb{Z}_3)$: intuitively, a position and momentum variable for each qutrit.  The contexts are in correspondence with the Lagrangian subspaces of $V$ whereas the stabilizer states can be computed from a Lagrangian subspace $M \subset V$ and a phase point $v \in V$ using \cite[Lemma 8]{gross2006hudson}.  The Lagrangian subspaces are easily enumerated by finding all linearly independent pairs of vectors and eliminating those with nonzero symplectic product.  Some pairs will generate the same subspace; this redundancy has to be eliminated.

To reach Result 1, the minimal independence numbers of the support graphs were computed to be equal to the number of contexts (40) and thus, by Theorem 2, the states are not logically contextual and thus not strongly contextual.  To reach Result 2, the independence number of the support graph of $\ket{CS}$ was computed and found to be 34.  Thus, by Theorem 3, $\ket{CS}$ is strongly contextual.

\subsection{Qubit stabilizer operations}
It is impossible to use mere contextuality as a sufficient criterion for determining which qubit states are suitable for MSD as the Peres-Mermin argument\footnotemark[\getrefnumber{proofB}] \cite{mermin} identifies all two-qubit states as contextual with respect to two-qubit stabilizer operations.  Howard \emph{et al.} posed as an open problem whether a quantification of contextuality could serve as a better criterion.  A close analysis of the Peres-Mermin argument shows that it, in fact, identifies all two-qubit states as \emph{strongly} contextual with respect to two-qubit stabilizer operations.  Thus, we suggest a negative answer to Howard \emph{et al.}'s question as such a quantification would need to discriminate between states which are all maximally contextual.  While it may still be possible to use contextuality as a criterion in identifying qubit MSD resources, it seems that such a measure will have to be specialized to the setting in question.   An interesting workaround to state-independent contextuality in qubit MSD is provided by \cite{raussrebit}.  

Howard \emph{et al.} describe noncontextuality with respect to all stabilizer operations in the sense of obeying all CSW inequalities on the orthogonality graph of stabilizer projectors.  In order to apply the AB contextuality criteria, it is necessary to describe this with a physical setting with a measurement scenario.  We chose as the set of measurements $\cm$ the 15 possible tensor pairs of $I,X,Y,Z$ excluding $I \otimes I$.  The contexts are the 15 maximal commuting subsets (of size 3) of these measurements.

Constructing the graph $\cg(\cm,\cc)$, we see that it is made up of 15 cliques of $2^3$ formal events: 120 formal observable events.  However, there are only 60 two-qubit stabilizer states.  The resolution of this puzzle is that half of the formal events represent joint outcomes which are actually impossible in quantum theory as they do not respect the algebraic relationships between the quantum measurements.

\begin{result}
All $n$-qubit states are strongly contextual with respect to $n$-qubit stabilizer operations whenever $n > 1$.\footnotemark[\getrefnumber{proofB}] 
\end{result}

Consider n = 2.  The exclusivity graph $\cg$ has as vertices the rank-1 two-qubit stabilizer projectors; edges connect orthogonal projectors.  We compute $\alpha(\cg) = 12$, agreeing with \cite[Table 3]{howard2013quantum}, and note that $\alpha(\cg_\rho) \leq \alpha(\cg)$ for the empirical model arising from any quantum state $\rho$.  Since the number of contexts is 15, it follows from Theorem 3 that $\rho$ is strongly contextual.  For $n > 2$, an analytic proof that directly extends the Peres-Mermin argument is given in Appendix B.

The above proof applies Theorem 3 for each quantum state individually by noting that the support graph of any quantum state is a subgraph of the support graph of the maximally mixed state.  Generalizing, we see that \emph{state-independent strong contextuality} will occur whenever the independence number of the orthogonality graph of a quantum measurement scenario is less than the number of contexts.  AB identified the connection between Kochen-Specker-type configurations of vectors and generic strong contextuality; the arguments of \cite[\S 7.1, \S 9.2]{AB} or the theory of generalized all-versus-nothing arguments \cite[\S 4.2]{ccp} can also be used to reach Result 3.


\section{Conclusions}
The above examples clarify the relationship between the logical strength of contextual resources and computational power.  We intend to further investigate the strength of contextuaity of higher-qudit magic states.  Further development will also require considering other measures of contextuality and investigating the relationship between measures of contextuality and other models of quantum computation.  Another important future direction is to clarify the relationship between the above described notion of contextuality with Spekkens' \cite{liang2011specker,spekkens} notion.

We have shown how to synthesize insights and tools from both the Cabello-Severini-Winter and Abramsky-Brandenburger approaches for understanding contextuality of operational data for very general experimental settings.  We repeatedly exploited the correspondence between canonical hidden variables for $(\cm,\cc)$ and independent sets of maximal size of $\cg(\cm,\cc)$ which relies on the special structure of exclusivity graphs arising from measurement scenarios.  An interesting future direction is understanding arbitrary graphs as \emph{contextual sample spaces}.  It will be necessary to understand how to complete an abstract graph to one arising from a measurement scenario.  Perhaps additional structure on the graph will be necessary for this.  For example, it may be necessary to fix a distinguished vertex cover by disjoint cliques indicating events of a common context; more general hypergraph-theoretic structure \cite{acin} might also be required.





\begin{acknowledgements}
The author wishes to thank Eric Cavalcanti, Samson Abramsky, Simone Severini, Rui Soares Barbosa, Kohei Kishida, Dan Browne, Shane Mansfield, Andrew Simmons, Mark Howard, and Ana Bel\'en Sainz for many helpful discussions. This research was financially supported by the EPSRC via grant EP/N017935/1 (Contextuality as a resource in quantum computation). This work was partially completed at the Simons Institute for the Theory 
of Computing (supported by the Simons Foundation) at 
the University of California, Berkeley, during the Logical Structures in Computation programme. 
\end{acknowledgements}


\bibliographystyle{apsrev4-1}
\bibliography{contextuality}

\begin{thebibliography}{56}%
\makeatletter
\providecommand \@ifxundefined [1]{%
 \@ifx{#1\undefined}
}%
\providecommand \@ifnum [1]{%
 \ifnum #1\expandafter \@firstoftwo
 \else \expandafter \@secondoftwo
 \fi
}%
\providecommand \@ifx [1]{%
 \ifx #1\expandafter \@firstoftwo
 \else \expandafter \@secondoftwo
 \fi
}%
\providecommand \natexlab [1]{#1}%
\providecommand \enquote  [1]{``#1''}%
\providecommand \bibnamefont  [1]{#1}%
\providecommand \bibfnamefont [1]{#1}%
\providecommand \citenamefont [1]{#1}%
\providecommand \href@noop [0]{\@secondoftwo}%
\providecommand \href [0]{\begingroup \@sanitize@url \@href}%
\providecommand \@href[1]{\@@startlink{#1}\@@href}%
\providecommand \@@href[1]{\endgroup#1\@@endlink}%
\providecommand \@sanitize@url [0]{\catcode `\\12\catcode `\$12\catcode
  `\&12\catcode `\#12\catcode `\^12\catcode `\_12\catcode `\%12\relax}%
\providecommand \@@startlink[1]{}%
\providecommand \@@endlink[0]{}%
\providecommand \url  [0]{\begingroup\@sanitize@url \@url }%
\providecommand \@url [1]{\endgroup\@href {#1}{\urlprefix }}%
\providecommand \urlprefix  [0]{URL }%
\providecommand \Eprint [0]{\href }%
\providecommand \doibase [0]{http://dx.doi.org/}%
\providecommand \selectlanguage [0]{\@gobble}%
\providecommand \bibinfo  [0]{\@secondoftwo}%
\providecommand \bibfield  [0]{\@secondoftwo}%
\providecommand \translation [1]{[#1]}%
\providecommand \BibitemOpen [0]{}%
\providecommand \bibitemStop [0]{}%
\providecommand \bibitemNoStop [0]{.\EOS\space}%
\providecommand \EOS [0]{\spacefactor3000\relax}%
\providecommand \BibitemShut  [1]{\csname bibitem#1\endcsname}%
\let\auto@bib@innerbib\@empty
\bibitem [{\citenamefont {Vidal}(2003)}]{vidal}%
  \BibitemOpen
  \bibfield  {author} {\bibinfo {author} {\bibfnamefont {G.}~\bibnamefont
  {Vidal}},\ }\href@noop {} {\bibfield  {journal} {\bibinfo  {journal} {Phys.
  Rev. Lett.}\ }\textbf {\bibinfo {volume} {91}},\ \bibinfo {pages} {147902}
  (\bibinfo {year} {2003})}\BibitemShut {NoStop}%
\bibitem [{\citenamefont {Deutsch}(1985)}]{deutsch}%
  \BibitemOpen
  \bibfield  {author} {\bibinfo {author} {\bibfnamefont {D.}~\bibnamefont
  {Deutsch}},\ }\href@noop {} {\bibfield  {journal} {\bibinfo  {journal} {Proc.
  R. Soc. A}\ }\textbf {\bibinfo {volume} {400}},\ \bibinfo {pages} {97}
  (\bibinfo {year} {1985})}\BibitemShut {NoStop}%
\bibitem [{\citenamefont {Datta}\ \emph {et~al.}(2008)\citenamefont {Datta},
  \citenamefont {Shaji},\ and\ \citenamefont {Caves}}]{datta}%
  \BibitemOpen
  \bibfield  {author} {\bibinfo {author} {\bibfnamefont {A.}~\bibnamefont
  {Datta}}, \bibinfo {author} {\bibfnamefont {A.}~\bibnamefont {Shaji}}, \ and\
  \bibinfo {author} {\bibfnamefont {C.~M.}\ \bibnamefont {Caves}},\ }\href@noop
  {} {\bibfield  {journal} {\bibinfo  {journal} {Phys. Rev. Lett.}\ }\textbf
  {\bibinfo {volume} {100}},\ \bibinfo {pages} {050502} (\bibinfo {year}
  {2008})}\BibitemShut {NoStop}%
\bibitem [{\citenamefont {Anders}\ and\ \citenamefont {Browne}(2009)}]{anders}%
  \BibitemOpen
  \bibfield  {author} {\bibinfo {author} {\bibfnamefont {J.}~\bibnamefont
  {Anders}}\ and\ \bibinfo {author} {\bibfnamefont {D.~E.}\ \bibnamefont
  {Browne}},\ }\href@noop {} {\bibfield  {journal} {\bibinfo  {journal} {Phys.
  Rev. Lett.}\ }\textbf {\bibinfo {volume} {102}},\ \bibinfo {pages} {050502}
  (\bibinfo {year} {2009})}\BibitemShut {NoStop}%
\bibitem [{\citenamefont {Galv\~ao}(2005)}]{galvao2005discrete}%
  \BibitemOpen
  \bibfield  {author} {\bibinfo {author} {\bibfnamefont {E.~F.}\ \bibnamefont
  {Galv\~ao}},\ }\href@noop {} {\bibfield  {journal} {\bibinfo  {journal}
  {Phys. Rev. A}\ }\textbf {\bibinfo {volume} {71}},\ \bibinfo {pages} {042302}
  (\bibinfo {year} {2005})}\BibitemShut {NoStop}%
\bibitem [{\citenamefont {Veitch}\ \emph {et~al.}(2012)\citenamefont {Veitch},
  \citenamefont {Ferrie}, \citenamefont {Gross},\ and\ \citenamefont
  {Emerson}}]{veitch}%
  \BibitemOpen
  \bibfield  {author} {\bibinfo {author} {\bibfnamefont {V.}~\bibnamefont
  {Veitch}}, \bibinfo {author} {\bibfnamefont {C.}~\bibnamefont {Ferrie}},
  \bibinfo {author} {\bibfnamefont {D.}~\bibnamefont {Gross}}, \ and\ \bibinfo
  {author} {\bibfnamefont {J.}~\bibnamefont {Emerson}},\ }\href@noop {}
  {\bibfield  {journal} {\bibinfo  {journal} {New J. Phys.}\ }\textbf {\bibinfo
  {volume} {14}},\ \bibinfo {pages} {113011} (\bibinfo {year}
  {2012})}\BibitemShut {NoStop}%
\bibitem [{\citenamefont {Delfosse}\ \emph {et~al.}(2015)\citenamefont
  {Delfosse}, \citenamefont {Allard~Guerin}, \citenamefont {Bian},\ and\
  \citenamefont {Raussendorf}}]{delfosse2015wigner}%
  \BibitemOpen
  \bibfield  {author} {\bibinfo {author} {\bibfnamefont {N.}~\bibnamefont
  {Delfosse}}, \bibinfo {author} {\bibfnamefont {P.}~\bibnamefont
  {Allard~Guerin}}, \bibinfo {author} {\bibfnamefont {J.}~\bibnamefont {Bian}},
  \ and\ \bibinfo {author} {\bibfnamefont {R.}~\bibnamefont {Raussendorf}},\
  }\href@noop {} {\bibfield  {journal} {\bibinfo  {journal} {Phys. Rev. X}\
  }\textbf {\bibinfo {volume} {5}},\ \bibinfo {pages} {021003} (\bibinfo {year}
  {2015})}\BibitemShut {NoStop}%
\bibitem [{\citenamefont {Raussendorf}\ \emph {et~al.}(2015)\citenamefont
  {Raussendorf}, \citenamefont {Browne}, \citenamefont {Delfosse},
  \citenamefont {Okay},\ and\ \citenamefont {Bermejo-Vega}}]{raussrebit}%
  \BibitemOpen
  \bibfield  {author} {\bibinfo {author} {\bibfnamefont {R.}~\bibnamefont
  {Raussendorf}}, \bibinfo {author} {\bibfnamefont {D.~E.}\ \bibnamefont
  {Browne}}, \bibinfo {author} {\bibfnamefont {N.}~\bibnamefont {Delfosse}},
  \bibinfo {author} {\bibfnamefont {C.}~\bibnamefont {Okay}}, \ and\ \bibinfo
  {author} {\bibfnamefont {J.}~\bibnamefont {Bermejo-Vega}},\ }\href@noop {}
  {\enquote {\bibinfo {title} {Contextuality as a resource for qubit quantum
  computation},}\ } (\bibinfo {year} {2015}),\ \Eprint
  {http://arxiv.org/abs/quant-ph/1511.08506} {quant-ph/1511.08506} \BibitemShut
  {NoStop}%
\bibitem [{\citenamefont {Howard}\ \emph {et~al.}(2014)\citenamefont {Howard},
  \citenamefont {Wallman}, \citenamefont {Veitch},\ and\ \citenamefont
  {Emerson}}]{howard}%
  \BibitemOpen
  \bibfield  {author} {\bibinfo {author} {\bibfnamefont {M.}~\bibnamefont
  {Howard}}, \bibinfo {author} {\bibfnamefont {J.}~\bibnamefont {Wallman}},
  \bibinfo {author} {\bibfnamefont {V.}~\bibnamefont {Veitch}}, \ and\ \bibinfo
  {author} {\bibfnamefont {J.}~\bibnamefont {Emerson}},\ }\href@noop {}
  {\bibfield  {journal} {\bibinfo  {journal} {Nature}\ }\textbf {\bibinfo
  {volume} {510}},\ \bibinfo {pages} {351} (\bibinfo {year}
  {2014})}\BibitemShut {NoStop}%
\bibitem [{\citenamefont {Raussendorf}(2013)}]{raussendorf}%
  \BibitemOpen
  \bibfield  {author} {\bibinfo {author} {\bibfnamefont {R.}~\bibnamefont
  {Raussendorf}},\ }\href@noop {} {\bibfield  {journal} {\bibinfo  {journal}
  {Phys. Rev. A}\ }\textbf {\bibinfo {volume} {88}},\ \bibinfo {pages} {022322}
  (\bibinfo {year} {2013})}\BibitemShut {NoStop}%
\bibitem [{\citenamefont {Bell}(1964)}]{bell1964einstein}%
  \BibitemOpen
  \bibfield  {author} {\bibinfo {author} {\bibfnamefont {J.~S.}\ \bibnamefont
  {Bell}},\ }\href@noop {} {\bibfield  {journal} {\bibinfo  {journal}
  {Physics}\ }\textbf {\bibinfo {volume} {1}},\ \bibinfo {pages} {195}
  (\bibinfo {year} {1964})}\BibitemShut {NoStop}%
\bibitem [{\citenamefont {Wiseman}(2014)}]{wiseman2014two}%
  \BibitemOpen
  \bibfield  {author} {\bibinfo {author} {\bibfnamefont {H.~M.}\ \bibnamefont
  {Wiseman}},\ }\href@noop {} {\bibfield  {journal} {\bibinfo  {journal} {J.
  Phys. A}\ }\textbf {\bibinfo {volume} {47}},\ \bibinfo {pages} {424001}
  (\bibinfo {year} {2014})}\BibitemShut {NoStop}%
\bibitem [{\citenamefont {Bennett}\ and\ \citenamefont
  {Wiesner}(1992)}]{benn1}%
  \BibitemOpen
  \bibfield  {author} {\bibinfo {author} {\bibfnamefont {C.}~\bibnamefont
  {Bennett}}\ and\ \bibinfo {author} {\bibfnamefont {S.~J.}\ \bibnamefont
  {Wiesner}},\ }\href@noop {} {\bibfield  {journal} {\bibinfo  {journal} {Phys.
  Rev. Lett.}\ }\textbf {\bibinfo {volume} {69}},\ \bibinfo {pages} {2881}
  (\bibinfo {year} {1992})}\BibitemShut {NoStop}%
\bibitem [{\citenamefont {Bennett}\ \emph {et~al.}(1993)\citenamefont
  {Bennett}, \citenamefont {Brassard}, \citenamefont {Cr\'epeau}, \citenamefont
  {Jozsa}, \citenamefont {Peres},\ and\ \citenamefont {Wootters}}]{benn2}%
  \BibitemOpen
  \bibfield  {author} {\bibinfo {author} {\bibfnamefont {C.~H.}\ \bibnamefont
  {Bennett}}, \bibinfo {author} {\bibfnamefont {G.}~\bibnamefont {Brassard}},
  \bibinfo {author} {\bibfnamefont {C.}~\bibnamefont {Cr\'epeau}}, \bibinfo
  {author} {\bibfnamefont {R.}~\bibnamefont {Jozsa}}, \bibinfo {author}
  {\bibfnamefont {A.}~\bibnamefont {Peres}}, \ and\ \bibinfo {author}
  {\bibfnamefont {W.~K.}\ \bibnamefont {Wootters}},\ }\href@noop {} {\bibfield
  {journal} {\bibinfo  {journal} {Phys. Rev. Lett.}\ }\textbf {\bibinfo
  {volume} {70}},\ \bibinfo {pages} {1895} (\bibinfo {year}
  {1993})}\BibitemShut {NoStop}%
\bibitem [{\citenamefont {Bennett}\ and\ \citenamefont
  {Brassard}(1984)}]{benn3}%
  \BibitemOpen
  \bibfield  {author} {\bibinfo {author} {\bibfnamefont {C.~H.}\ \bibnamefont
  {Bennett}}\ and\ \bibinfo {author} {\bibfnamefont {G.}~\bibnamefont
  {Brassard}},\ }in\ \href@noop {} {\emph {\bibinfo {booktitle} {International
  Conference on Computer System and Signal Processing, IEEE, 1984}}}\ (\bibinfo
  {year} {1984})\ pp.\ \bibinfo {pages} {175--179}\BibitemShut {NoStop}%
\bibitem [{\citenamefont {Hardy}(1992)}]{hardy1992quantum}%
  \BibitemOpen
  \bibfield  {author} {\bibinfo {author} {\bibfnamefont {L.}~\bibnamefont
  {Hardy}},\ }\href@noop {} {\bibfield  {journal} {\bibinfo  {journal} {Phys.
  Rev. Lett.}\ }\textbf {\bibinfo {volume} {68}},\ \bibinfo {pages} {2981}
  (\bibinfo {year} {1992})}\BibitemShut {NoStop}%
\bibitem [{\citenamefont {Hardy}(1993)}]{hardy1993}%
  \BibitemOpen
  \bibfield  {author} {\bibinfo {author} {\bibfnamefont {L.}~\bibnamefont
  {Hardy}},\ }\href@noop {} {\bibfield  {journal} {\bibinfo  {journal} {Phys.
  Rev. Lett.}\ }\textbf {\bibinfo {volume} {71}},\ \bibinfo {pages} {1665}
  (\bibinfo {year} {1993})}\BibitemShut {NoStop}%
\bibitem [{\citenamefont {Greenberger}\ \emph {et~al.}(1990)\citenamefont
  {Greenberger}, \citenamefont {Horne}, \citenamefont {Shimony},\ and\
  \citenamefont {Zeilinger}}]{ghz}%
  \BibitemOpen
  \bibfield  {author} {\bibinfo {author} {\bibfnamefont {D.~M.}\ \bibnamefont
  {Greenberger}}, \bibinfo {author} {\bibfnamefont {M.~A.}\ \bibnamefont
  {Horne}}, \bibinfo {author} {\bibfnamefont {A.}~\bibnamefont {Shimony}}, \
  and\ \bibinfo {author} {\bibfnamefont {A.}~\bibnamefont {Zeilinger}},\
  }\href@noop {} {\bibfield  {journal} {\bibinfo  {journal} {Am. J. Phys}\
  }\textbf {\bibinfo {volume} {58}},\ \bibinfo {pages} {1131} (\bibinfo {year}
  {1990})}\BibitemShut {NoStop}%
\bibitem [{\citenamefont {van Dam}(1999)}]{vandamthesis}%
  \BibitemOpen
  \bibfield  {author} {\bibinfo {author} {\bibfnamefont {W.}~\bibnamefont {van
  Dam}},\ }\emph {\bibinfo {title} {Nonlocality \& communication complexity}},\
  \href@noop {} {Ph.D. thesis},\ \bibinfo  {school} {Faculty of Physical
  Sciences, University of Oxford} (\bibinfo {year} {1999})\BibitemShut
  {NoStop}%
\bibitem [{\citenamefont {Barrett}(2007)}]{barrett}%
  \BibitemOpen
  \bibfield  {author} {\bibinfo {author} {\bibfnamefont {J.}~\bibnamefont
  {Barrett}},\ }\href@noop {} {\bibfield  {journal} {\bibinfo  {journal} {Phys.
  Rev. A}\ }\textbf {\bibinfo {volume} {75}},\ \bibinfo {pages} {032304}
  (\bibinfo {year} {2007})}\BibitemShut {NoStop}%
\bibitem [{\citenamefont {Barrett}\ \emph {et~al.}(2005)\citenamefont
  {Barrett}, \citenamefont {Linden}, \citenamefont {Massar}, \citenamefont
  {Pironio}, \citenamefont {Popescu},\ and\ \citenamefont
  {Roberts}}]{barrett2005nonlocal}%
  \BibitemOpen
  \bibfield  {author} {\bibinfo {author} {\bibfnamefont {J.}~\bibnamefont
  {Barrett}}, \bibinfo {author} {\bibfnamefont {N.}~\bibnamefont {Linden}},
  \bibinfo {author} {\bibfnamefont {S.}~\bibnamefont {Massar}}, \bibinfo
  {author} {\bibfnamefont {S.}~\bibnamefont {Pironio}}, \bibinfo {author}
  {\bibfnamefont {S.}~\bibnamefont {Popescu}}, \ and\ \bibinfo {author}
  {\bibfnamefont {D.}~\bibnamefont {Roberts}},\ }\href@noop {} {\bibfield
  {journal} {\bibinfo  {journal} {Phys. Rev. A}\ }\textbf {\bibinfo {volume}
  {71}},\ \bibinfo {pages} {022101} (\bibinfo {year} {2005})}\BibitemShut
  {NoStop}%
\bibitem [{\citenamefont {Buhrman}\ \emph {et~al.}(2010)\citenamefont
  {Buhrman}, \citenamefont {Cleve}, \citenamefont {Massar},\ and\ \citenamefont
  {de~Wolf}}]{buhrman2010nonlocality}%
  \BibitemOpen
  \bibfield  {author} {\bibinfo {author} {\bibfnamefont {H.}~\bibnamefont
  {Buhrman}}, \bibinfo {author} {\bibfnamefont {R.}~\bibnamefont {Cleve}},
  \bibinfo {author} {\bibfnamefont {S.}~\bibnamefont {Massar}}, \ and\ \bibinfo
  {author} {\bibfnamefont {R.}~\bibnamefont {de~Wolf}},\ }\href@noop {}
  {\bibfield  {journal} {\bibinfo  {journal} {Rev. Mod. Phys.}\ }\textbf
  {\bibinfo {volume} {82}},\ \bibinfo {pages} {665} (\bibinfo {year}
  {2010})}\BibitemShut {NoStop}%
\bibitem [{\citenamefont {Kochen}\ and\ \citenamefont {Specker}(1967)}]{ks}%
  \BibitemOpen
  \bibfield  {author} {\bibinfo {author} {\bibfnamefont {S.}~\bibnamefont
  {Kochen}}\ and\ \bibinfo {author} {\bibfnamefont {E.~P.}\ \bibnamefont
  {Specker}},\ }\href@noop {} {\bibfield  {journal} {\bibinfo  {journal} {J.
  Math. Mech.}\ }\textbf {\bibinfo {volume} {17}},\ \bibinfo {pages} {59}
  (\bibinfo {year} {1967})}\BibitemShut {NoStop}%
\bibitem [{\citenamefont {Fine}(1982)}]{fine}%
  \BibitemOpen
  \bibfield  {author} {\bibinfo {author} {\bibfnamefont {A.}~\bibnamefont
  {Fine}},\ }\href@noop {} {\bibfield  {journal} {\bibinfo  {journal} {Phys.
  Rev. Lett.}\ }\textbf {\bibinfo {volume} {48}},\ \bibinfo {pages} {291}
  (\bibinfo {year} {1982})}\BibitemShut {NoStop}%
\bibitem [{\citenamefont {Klyachko}\ \emph {et~al.}(2008)\citenamefont
  {Klyachko}, \citenamefont {Can}, \citenamefont {Binicio\u{g}lu},\ and\
  \citenamefont {Shumovsky}}]{klyachko}%
  \BibitemOpen
  \bibfield  {author} {\bibinfo {author} {\bibfnamefont {A.~A.}\ \bibnamefont
  {Klyachko}}, \bibinfo {author} {\bibfnamefont {M.~A.}\ \bibnamefont {Can}},
  \bibinfo {author} {\bibfnamefont {S.}~\bibnamefont {Binicio\u{g}lu}}, \ and\
  \bibinfo {author} {\bibfnamefont {A.~S.}\ \bibnamefont {Shumovsky}},\
  }\href@noop {} {\bibfield  {journal} {\bibinfo  {journal} {Phys. Rev. Lett.}\
  }\textbf {\bibinfo {volume} {101}},\ \bibinfo {pages} {020403} (\bibinfo
  {year} {2008})}\BibitemShut {NoStop}%
\bibitem [{\citenamefont {Bravyi}\ and\ \citenamefont
  {Kitaev}(2005)}]{bravkit}%
  \BibitemOpen
  \bibfield  {author} {\bibinfo {author} {\bibfnamefont {S.}~\bibnamefont
  {Bravyi}}\ and\ \bibinfo {author} {\bibfnamefont {A.}~\bibnamefont
  {Kitaev}},\ }\href@noop {} {\bibfield  {journal} {\bibinfo  {journal} {Phys.
  Rev. A}\ }\textbf {\bibinfo {volume} {71}},\ \bibinfo {pages} {022316}
  (\bibinfo {year} {2005})}\BibitemShut {NoStop}%
\bibitem [{\citenamefont {Raussendorf}\ and\ \citenamefont
  {Briegel}(2001)}]{raussendorf2001one}%
  \BibitemOpen
  \bibfield  {author} {\bibinfo {author} {\bibfnamefont {R.}~\bibnamefont
  {Raussendorf}}\ and\ \bibinfo {author} {\bibfnamefont {H.~J.}\ \bibnamefont
  {Briegel}},\ }\href@noop {} {\bibfield  {journal} {\bibinfo  {journal} {Phys.
  Rev. Lett.}\ }\textbf {\bibinfo {volume} {86}},\ \bibinfo {pages} {5188}
  (\bibinfo {year} {2001})}\BibitemShut {NoStop}%
\bibitem [{\citenamefont {Cabello}\ \emph {et~al.}(2014)\citenamefont
  {Cabello}, \citenamefont {Severini},\ and\ \citenamefont {Winter}}]{CSW}%
  \BibitemOpen
  \bibfield  {author} {\bibinfo {author} {\bibfnamefont {A.}~\bibnamefont
  {Cabello}}, \bibinfo {author} {\bibfnamefont {S.}~\bibnamefont {Severini}}, \
  and\ \bibinfo {author} {\bibfnamefont {A.}~\bibnamefont {Winter}},\
  }\href@noop {} {\bibfield  {journal} {\bibinfo  {journal} {Phys. Rev. Lett.}\
  }\textbf {\bibinfo {volume} {112}},\ \bibinfo {pages} {040401} (\bibinfo
  {year} {2014})}\BibitemShut {NoStop}%
\bibitem [{\citenamefont {Abramsky}\ and\ \citenamefont
  {Brandenburger}(2011)}]{AB}%
  \BibitemOpen
  \bibfield  {author} {\bibinfo {author} {\bibfnamefont {S.}~\bibnamefont
  {Abramsky}}\ and\ \bibinfo {author} {\bibfnamefont {A.}~\bibnamefont
  {Brandenburger}},\ }\href@noop {} {\bibfield  {journal} {\bibinfo  {journal}
  {New J. Phys.}\ }\textbf {\bibinfo {volume} {13}},\ \bibinfo {pages} {113036}
  (\bibinfo {year} {2011})}\BibitemShut {NoStop}%
\bibitem [{\citenamefont {Howard}\ and\ \citenamefont
  {Campbell}(2016)}]{howard2016application}%
  \BibitemOpen
  \bibfield  {author} {\bibinfo {author} {\bibfnamefont {M.}~\bibnamefont
  {Howard}}\ and\ \bibinfo {author} {\bibfnamefont {E.~T.}\ \bibnamefont
  {Campbell}},\ }\href@noop {} {\  (\bibinfo {year} {2016})},\ \Eprint
  {http://arxiv.org/abs/quant-ph/1609.07488} {quant-ph/1609.07488} \BibitemShut
  {NoStop}%
\bibitem [{\citenamefont {Abramsky}\ and\ \citenamefont
  {Hardy}(2012)}]{logicalbi}%
  \BibitemOpen
  \bibfield  {author} {\bibinfo {author} {\bibfnamefont {S.}~\bibnamefont
  {Abramsky}}\ and\ \bibinfo {author} {\bibfnamefont {L.}~\bibnamefont
  {Hardy}},\ }\href@noop {} {\bibfield  {journal} {\bibinfo  {journal} {Phys.
  Rev. A}\ }\textbf {\bibinfo {volume} {85}},\ \bibinfo {pages} {062114}
  (\bibinfo {year} {2012})}\BibitemShut {NoStop}%
\bibitem [{\citenamefont {Ac{\'\i}n}\ \emph {et~al.}(2015)\citenamefont
  {Ac{\'\i}n}, \citenamefont {Fritz}, \citenamefont {Leverrier},\ and\
  \citenamefont {Sainz}}]{acin}%
  \BibitemOpen
  \bibfield  {author} {\bibinfo {author} {\bibfnamefont {A.}~\bibnamefont
  {Ac{\'\i}n}}, \bibinfo {author} {\bibfnamefont {T.}~\bibnamefont {Fritz}},
  \bibinfo {author} {\bibfnamefont {A.}~\bibnamefont {Leverrier}}, \ and\
  \bibinfo {author} {\bibfnamefont {A.~B.}\ \bibnamefont {Sainz}},\ }\href@noop
  {} {\bibfield  {journal} {\bibinfo  {journal} {Comm. Math. Phys.}\ }\textbf
  {\bibinfo {volume} {334}},\ \bibinfo {pages} {533} (\bibinfo {year}
  {2015})}\BibitemShut {NoStop}%
\bibitem [{\citenamefont {Clauser}\ \emph {et~al.}(1969)\citenamefont
  {Clauser}, \citenamefont {Horne}, \citenamefont {Shimony},\ and\
  \citenamefont {Holt}}]{chsh}%
  \BibitemOpen
  \bibfield  {author} {\bibinfo {author} {\bibfnamefont {J.~F.}\ \bibnamefont
  {Clauser}}, \bibinfo {author} {\bibfnamefont {M.~A.}\ \bibnamefont {Horne}},
  \bibinfo {author} {\bibfnamefont {A.}~\bibnamefont {Shimony}}, \ and\
  \bibinfo {author} {\bibfnamefont {R.~A.}\ \bibnamefont {Holt}},\ }\href@noop
  {} {\bibfield  {journal} {\bibinfo  {journal} {Phys. Rev. Lett.}\ }\textbf
  {\bibinfo {volume} {23}},\ \bibinfo {pages} {880} (\bibinfo {year}
  {1969})}\BibitemShut {NoStop}%
\bibitem [{\citenamefont {Popescu}\ and\ \citenamefont
  {Rohrlich}(1994)}]{popescu1994quantum}%
  \BibitemOpen
  \bibfield  {author} {\bibinfo {author} {\bibfnamefont {S.}~\bibnamefont
  {Popescu}}\ and\ \bibinfo {author} {\bibfnamefont {D.}~\bibnamefont
  {Rohrlich}},\ }\href@noop {} {\bibfield  {journal} {\bibinfo  {journal}
  {Found. Phys.}\ }\textbf {\bibinfo {volume} {24}},\ \bibinfo {pages} {379}
  (\bibinfo {year} {1994})}\BibitemShut {NoStop}%
\bibitem [{\citenamefont {Sadiq}\ \emph {et~al.}(2013)\citenamefont {Sadiq},
  \citenamefont {Badziag}, \citenamefont {Bourennane},\ and\ \citenamefont
  {Cabello}}]{sadiq}%
  \BibitemOpen
  \bibfield  {author} {\bibinfo {author} {\bibfnamefont {M.}~\bibnamefont
  {Sadiq}}, \bibinfo {author} {\bibfnamefont {P.}~\bibnamefont {Badziag}},
  \bibinfo {author} {\bibfnamefont {M.}~\bibnamefont {Bourennane}}, \ and\
  \bibinfo {author} {\bibfnamefont {A.}~\bibnamefont {Cabello}},\ }\href@noop
  {} {\bibfield  {journal} {\bibinfo  {journal} {Phys. Rev. A}\ }\textbf
  {\bibinfo {volume} {87}},\ \bibinfo {pages} {012128} (\bibinfo {year}
  {2013})}\BibitemShut {NoStop}%
\bibitem [{\citenamefont {Fritz}\ \emph {et~al.}(2013)\citenamefont {Fritz},
  \citenamefont {Sainz}, \citenamefont {Augusiak}, \citenamefont {Brask},
  \citenamefont {Chaves}, \citenamefont {Leverrier},\ and\ \citenamefont
  {Ac{\'\i}n}}]{fritz2013local}%
  \BibitemOpen
  \bibfield  {author} {\bibinfo {author} {\bibfnamefont {T.}~\bibnamefont
  {Fritz}}, \bibinfo {author} {\bibfnamefont {A.~B.}\ \bibnamefont {Sainz}},
  \bibinfo {author} {\bibfnamefont {R.}~\bibnamefont {Augusiak}}, \bibinfo
  {author} {\bibfnamefont {J.~B.}\ \bibnamefont {Brask}}, \bibinfo {author}
  {\bibfnamefont {R.}~\bibnamefont {Chaves}}, \bibinfo {author} {\bibfnamefont
  {A.}~\bibnamefont {Leverrier}}, \ and\ \bibinfo {author} {\bibfnamefont
  {A.}~\bibnamefont {Ac{\'\i}n}},\ }\href@noop {} {\bibfield  {journal}
  {\bibinfo  {journal} {Nature Comm.}\ }\textbf {\bibinfo {volume} {4}}
  (\bibinfo {year} {2013})}\BibitemShut {NoStop}%
\bibitem [{\citenamefont {Sainz}\ \emph {et~al.}(2014)\citenamefont {Sainz},
  \citenamefont {Fritz}, \citenamefont {Augusiak}, \citenamefont {Brask},
  \citenamefont {Chaves}, \citenamefont {Leverrier},\ and\ \citenamefont
  {Ac{\'\i}n}}]{sainz2014exploring}%
  \BibitemOpen
  \bibfield  {author} {\bibinfo {author} {\bibfnamefont {A.~B.}\ \bibnamefont
  {Sainz}}, \bibinfo {author} {\bibfnamefont {T.}~\bibnamefont {Fritz}},
  \bibinfo {author} {\bibfnamefont {R.}~\bibnamefont {Augusiak}}, \bibinfo
  {author} {\bibfnamefont {J.~B.}\ \bibnamefont {Brask}}, \bibinfo {author}
  {\bibfnamefont {R.}~\bibnamefont {Chaves}}, \bibinfo {author} {\bibfnamefont
  {A.}~\bibnamefont {Leverrier}}, \ and\ \bibinfo {author} {\bibfnamefont
  {A.}~\bibnamefont {Ac{\'\i}n}},\ }\href@noop {} {\bibfield  {journal}
  {\bibinfo  {journal} {Phys. Rev. A}\ }\textbf {\bibinfo {volume} {89}},\
  \bibinfo {pages} {032117} (\bibinfo {year} {2014})}\BibitemShut {NoStop}%
\bibitem [{\citenamefont {Cabello}(2001)}]{cabello2001bell}%
  \BibitemOpen
  \bibfield  {author} {\bibinfo {author} {\bibfnamefont {A.}~\bibnamefont
  {Cabello}},\ }\href@noop {} {\bibfield  {journal} {\bibinfo  {journal} {Phys.
  Rev. Lett.}\ }\textbf {\bibinfo {volume} {86}},\ \bibinfo {pages} {1911}
  (\bibinfo {year} {2001})}\BibitemShut {NoStop}%
\bibitem [{\citenamefont {Cabello}(2002)}]{cabello2002bell}%
  \BibitemOpen
  \bibfield  {author} {\bibinfo {author} {\bibfnamefont {A.}~\bibnamefont
  {Cabello}},\ }\href@noop {} {\bibfield  {journal} {\bibinfo  {journal} {Phys.
  Rev. A}\ }\textbf {\bibinfo {volume} {65}},\ \bibinfo {pages} {032108}
  (\bibinfo {year} {2002})}\BibitemShut {NoStop}%
\bibitem [{\citenamefont {Yokota}\ \emph {et~al.}(2009)\citenamefont {Yokota},
  \citenamefont {Yamamoto}, \citenamefont {Koashi},\ and\ \citenamefont
  {Imoto}}]{yokota2009direct}%
  \BibitemOpen
  \bibfield  {author} {\bibinfo {author} {\bibfnamefont {K.}~\bibnamefont
  {Yokota}}, \bibinfo {author} {\bibfnamefont {T.}~\bibnamefont {Yamamoto}},
  \bibinfo {author} {\bibfnamefont {M.}~\bibnamefont {Koashi}}, \ and\ \bibinfo
  {author} {\bibfnamefont {N.}~\bibnamefont {Imoto}},\ }\href@noop {}
  {\bibfield  {journal} {\bibinfo  {journal} {New J. Phys.}\ }\textbf {\bibinfo
  {volume} {11}},\ \bibinfo {pages} {033011} (\bibinfo {year}
  {2009})}\BibitemShut {NoStop}%
\bibitem [{\citenamefont {Aharonov}\ \emph {et~al.}(2002)\citenamefont
  {Aharonov}, \citenamefont {Botero}, \citenamefont {Popescu}, \citenamefont
  {Reznik},\ and\ \citenamefont {Tollaksen}}]{aharonov2002revisiting}%
  \BibitemOpen
  \bibfield  {author} {\bibinfo {author} {\bibfnamefont {Y.}~\bibnamefont
  {Aharonov}}, \bibinfo {author} {\bibfnamefont {A.}~\bibnamefont {Botero}},
  \bibinfo {author} {\bibfnamefont {S.}~\bibnamefont {Popescu}}, \bibinfo
  {author} {\bibfnamefont {B.}~\bibnamefont {Reznik}}, \ and\ \bibinfo {author}
  {\bibfnamefont {J.}~\bibnamefont {Tollaksen}},\ }\href@noop {} {\bibfield
  {journal} {\bibinfo  {journal} {Phys. Lett. A}\ }\textbf {\bibinfo {volume}
  {301}},\ \bibinfo {pages} {130} (\bibinfo {year} {2002})}\BibitemShut
  {NoStop}%
\bibitem [{\citenamefont {Elitzur}\ \emph {et~al.}(1992)\citenamefont
  {Elitzur}, \citenamefont {Popescu},\ and\ \citenamefont {Rohrlich}}]{eprohr}%
  \BibitemOpen
  \bibfield  {author} {\bibinfo {author} {\bibfnamefont {A.~C.}\ \bibnamefont
  {Elitzur}}, \bibinfo {author} {\bibfnamefont {S.}~\bibnamefont {Popescu}}, \
  and\ \bibinfo {author} {\bibfnamefont {D.}~\bibnamefont {Rohrlich}},\
  }\href@noop {} {\bibfield  {journal} {\bibinfo  {journal} {Phys. Lett. A}\
  }\textbf {\bibinfo {volume} {162}},\ \bibinfo {pages} {25} (\bibinfo {year}
  {1992})}\BibitemShut {NoStop}%
\bibitem [{\citenamefont {Barrett}\ \emph {et~al.}(2006)\citenamefont
  {Barrett}, \citenamefont {Kent},\ and\ \citenamefont {Pironio}}]{barrettmax}%
  \BibitemOpen
  \bibfield  {author} {\bibinfo {author} {\bibfnamefont {J.}~\bibnamefont
  {Barrett}}, \bibinfo {author} {\bibfnamefont {A.}~\bibnamefont {Kent}}, \
  and\ \bibinfo {author} {\bibfnamefont {S.}~\bibnamefont {Pironio}},\
  }\href@noop {} {\bibfield  {journal} {\bibinfo  {journal} {Phys. Rev. Lett.}\
  }\textbf {\bibinfo {volume} {97}},\ \bibinfo {pages} {170409} (\bibinfo
  {year} {2006})}\BibitemShut {NoStop}%
\bibitem [{\citenamefont {Abramsky}\ \emph {et~al.}(2016)\citenamefont
  {Abramsky}, \citenamefont {Barbosa}, \citenamefont {Kishida}, \citenamefont
  {R.},\ and\ \citenamefont {Mansfield}}]{posspoly}%
  \BibitemOpen
  \bibfield  {author} {\bibinfo {author} {\bibfnamefont {S.}~\bibnamefont
  {Abramsky}}, \bibinfo {author} {\bibfnamefont {R.~S.}\ \bibnamefont
  {Barbosa}}, \bibinfo {author} {\bibfnamefont {K.}~\bibnamefont {Kishida}},
  \bibinfo {author} {\bibfnamefont {L.}~\bibnamefont {R.}}, \ and\ \bibinfo
  {author} {\bibfnamefont {S.}~\bibnamefont {Mansfield}},\ }\href@noop {}
  {\enquote {\bibinfo {title} {Possibilities determine the combinatorial
  structure of probability polytopes},}\ } (\bibinfo {year} {2016}),\ \Eprint
  {http://arxiv.org/abs/quant-ph/1603.07735} {quant-ph/1603.07735} \BibitemShut
  {NoStop}%
\bibitem [{\citenamefont {Gottesman}(1997)}]{gotthesis}%
  \BibitemOpen
  \bibfield  {author} {\bibinfo {author} {\bibfnamefont {D.}~\bibnamefont
  {Gottesman}},\ }\emph {\bibinfo {title} {Stabilizer codes and quantum error
  correction}},\ \href@noop {} {Ph.D. thesis},\ \bibinfo  {school} {Caltech}
  (\bibinfo {year} {1997})\BibitemShut {NoStop}%
\bibitem [{\citenamefont {Gross}(2006)}]{gross2006hudson}%
  \BibitemOpen
  \bibfield  {author} {\bibinfo {author} {\bibfnamefont {D.}~\bibnamefont
  {Gross}},\ }\href@noop {} {\bibfield  {journal} {\bibinfo  {journal} {J.
  Math. Phys.}\ }\textbf {\bibinfo {volume} {47}},\ \bibinfo {pages} {122107}
  (\bibinfo {year} {2006})}\BibitemShut {NoStop}%
\bibitem [{\citenamefont {Anwar}(2014)}]{anwar}%
  \BibitemOpen
  \bibfield  {author} {\bibinfo {author} {\bibfnamefont {H.}~\bibnamefont
  {Anwar}},\ }\emph {\bibinfo {title} {Towards fault-tolerant quantum
  computation with higher-dimensional systems}},\ \href@noop {} {Ph.D.
  thesis},\ \bibinfo  {school} {University College London} (\bibinfo {year}
  {2014})\BibitemShut {NoStop}%
\bibitem [{\citenamefont {Dawkins}\ and\ \citenamefont
  {Howard}(2015)}]{hillary}%
  \BibitemOpen
  \bibfield  {author} {\bibinfo {author} {\bibfnamefont {H.}~\bibnamefont
  {Dawkins}}\ and\ \bibinfo {author} {\bibfnamefont {M.}~\bibnamefont
  {Howard}},\ }\href@noop {} {\bibfield  {journal} {\bibinfo  {journal} {Phys.
  Rev. Lett.}\ }\textbf {\bibinfo {volume} {115}},\ \bibinfo {pages} {030501}
  (\bibinfo {year} {2015})}\BibitemShut {NoStop}%
\bibitem [{\citenamefont {Howard}\ and\ \citenamefont {Vala}(2012)}]{jiri}%
  \BibitemOpen
  \bibfield  {author} {\bibinfo {author} {\bibfnamefont {M.}~\bibnamefont
  {Howard}}\ and\ \bibinfo {author} {\bibfnamefont {J.}~\bibnamefont {Vala}},\
  }\href@noop {} {\bibfield  {journal} {\bibinfo  {journal} {Phys. Rev. A}\
  }\textbf {\bibinfo {volume} {86}},\ \bibinfo {pages} {022316} (\bibinfo
  {year} {2012})}\BibitemShut {NoStop}%
\bibitem [{\citenamefont {Howard}()}]{markpc}%
  \BibitemOpen
  \bibfield  {author} {\bibinfo {author} {\bibfnamefont {M.}~\bibnamefont
  {Howard}},\ }\href@noop {} {}\bibinfo {howpublished} {personal
  communication}\BibitemShut {NoStop}%
\bibitem [{\citenamefont {Cubitt}\ \emph {et~al.}(2010)\citenamefont {Cubitt},
  \citenamefont {Leung}, \citenamefont {Matthews},\ and\ \citenamefont
  {Winter}}]{cubitt}%
  \BibitemOpen
  \bibfield  {author} {\bibinfo {author} {\bibfnamefont {T.~S.}\ \bibnamefont
  {Cubitt}}, \bibinfo {author} {\bibfnamefont {D.}~\bibnamefont {Leung}},
  \bibinfo {author} {\bibfnamefont {W.}~\bibnamefont {Matthews}}, \ and\
  \bibinfo {author} {\bibfnamefont {A.}~\bibnamefont {Winter}},\ }\href@noop {}
  {\bibfield  {journal} {\bibinfo  {journal} {Phys. Rev. Lett.}\ }\textbf
  {\bibinfo {volume} {104}},\ \bibinfo {pages} {230503} (\bibinfo {year}
  {2010})}\BibitemShut {NoStop}%
\bibitem [{\citenamefont {Mermin}(1993)}]{mermin}%
  \BibitemOpen
  \bibfield  {author} {\bibinfo {author} {\bibfnamefont {N.~D.}\ \bibnamefont
  {Mermin}},\ }\href@noop {} {\bibfield  {journal} {\bibinfo  {journal} {Rev.
  Mod. Phys.}\ }\textbf {\bibinfo {volume} {65}},\ \bibinfo {pages} {803}
  (\bibinfo {year} {1993})}\BibitemShut {NoStop}%
\bibitem [{\citenamefont {Howard}\ \emph {et~al.}(2013)\citenamefont {Howard},
  \citenamefont {Brennan},\ and\ \citenamefont {Vala}}]{howard2013quantum}%
  \BibitemOpen
  \bibfield  {author} {\bibinfo {author} {\bibfnamefont {M.}~\bibnamefont
  {Howard}}, \bibinfo {author} {\bibfnamefont {E.}~\bibnamefont {Brennan}}, \
  and\ \bibinfo {author} {\bibfnamefont {J.}~\bibnamefont {Vala}},\ }\href@noop
  {} {\bibfield  {journal} {\bibinfo  {journal} {Entropy}\ }\textbf {\bibinfo
  {volume} {15}},\ \bibinfo {pages} {2340} (\bibinfo {year}
  {2013})}\BibitemShut {NoStop}%
\bibitem [{\citenamefont {{Abramsky}}\ \emph {et~al.}(2016)\citenamefont
  {{Abramsky}}, \citenamefont {{Soares Barbosa}}, \citenamefont {{Kishida}},
  \citenamefont {{Lal}},\ and\ \citenamefont {{Mansfield}}}]{ccp}%
  \BibitemOpen
  \bibfield  {author} {\bibinfo {author} {\bibfnamefont {S.}~\bibnamefont
  {{Abramsky}}}, \bibinfo {author} {\bibfnamefont {R.}~\bibnamefont {{Soares
  Barbosa}}}, \bibinfo {author} {\bibfnamefont {K.}~\bibnamefont {{Kishida}}},
  \bibinfo {author} {\bibfnamefont {R.}~\bibnamefont {{Lal}}}, \ and\ \bibinfo
  {author} {\bibfnamefont {S.}~\bibnamefont {{Mansfield}}},\ }\href@noop {}
  {\enquote {\bibinfo {title} {{Contextuality, cohomology and paradox}},}\ }
  (\bibinfo {year} {2016}),\ \Eprint {http://arxiv.org/abs/quant-ph/1502.03097}
  {quant-ph/1502.03097} \BibitemShut {NoStop}%
\bibitem [{\citenamefont {Liang}\ \emph {et~al.}(2011)\citenamefont {Liang},
  \citenamefont {Spekkens},\ and\ \citenamefont {Wiseman}}]{liang2011specker}%
  \BibitemOpen
  \bibfield  {author} {\bibinfo {author} {\bibfnamefont {Y.-C.}\ \bibnamefont
  {Liang}}, \bibinfo {author} {\bibfnamefont {R.~W.}\ \bibnamefont {Spekkens}},
  \ and\ \bibinfo {author} {\bibfnamefont {H.~M.}\ \bibnamefont {Wiseman}},\
  }\href@noop {} {\bibfield  {journal} {\bibinfo  {journal} {Phys. Rep.}\
  }\textbf {\bibinfo {volume} {506}},\ \bibinfo {pages} {1} (\bibinfo {year}
  {2011})}\BibitemShut {NoStop}%
\bibitem [{\citenamefont {Spekkens}(2005)}]{spekkens}%
  \BibitemOpen
  \bibfield  {author} {\bibinfo {author} {\bibfnamefont {R.~W.}\ \bibnamefont
  {Spekkens}},\ }\href@noop {} {\bibfield  {journal} {\bibinfo  {journal}
  {Phys. Rev. A}\ }\textbf {\bibinfo {volume} {71}},\ \bibinfo {pages} {052108}
  (\bibinfo {year} {2005})}\BibitemShut {NoStop}%
\end{thebibliography}%

\begin{appendices}

\section{Background}

\subsection{Graph theory}

A (simple, undirected) \emph{graph} $G$ is a set of vertices $V = \{v_1, ..., v_n\}$, pairs of which may be adjacent, i.e joined by an edge.  The edges are represented by a set $E \subset V \times V$ such that $(v,v) \notin E$ for any vertex $v$ and such that $(v,v') \in E$ if and only if $(v',v) \in E$.  That is, no edge joins a vertex to itself and $v$ and $v'$ are adjacent whenever $v'$ and $v$ are.  Given a subset $V' \subset V$ of vertices, the \emph{induced subgraph} has $V'$ as vertices and retains all edges $(v,v') \in E$ with $v,v' \in V'$.

A weighted graph $(G,w)$ is a graph $G$ together with a function $w: V \to \mathbb{R}^{\geq 0}$ from vertices to non-negative reals.  We denote the weight $w(v_i)$ of a vertex $v_i$ by $w_i$.

An \emph{independent (vertex) set} is a subset $I \subset V$ such that no two vertices in $I$ are adjacent: $(v, v') \notin E$ whenever $v, v' \in I$.  A \emph{clique} is a subset $C \subset V$ such that every two distinct vertices in $C$ are adjacent: $(v, v') \in E$ whenever $v, v' \in C$  and $v \neq v'$.

The \emph{independence number} of a graph $G$ is the size of the largest independent set in $G$; in terms of Definition 3, it is the maximal independence degree over all vertices.  The independence number of a weighted graph $(G,w)$ is the maximum value of the sum $\sum_{i \in I} w_i$ where $I$ is any independent set.

To define the \emph{Lov\'asz theta number} $\vartheta(G,w)$ of a weighted graph $(G,w)$ we must first define an orthonormal representation of a graph $G$.  This is a choice of a unit vector $\ket{\phi} \in \mathbb{R}^d$ and an assignment to each vertex $v_i$ a unit vector $\ket{\psi_i} \in \mathbb{R}^d$ such that $\braket{\psi_i|\psi_j} = 0$  whenever $v_i$ and $v_j$ are adjacent.  The choice of dimension $d$ can be arbitrary.  The Lov\'asz theta number $\vartheta(G,w)$ is then defined to be the maximum $\sum_{i \in V} w_i |\braket{\psi_i|\phi}|^2$ over all possible orthonormal representations of $G$.

The \emph{fractional packing number} $\alpha^*(G,w)$ is the maximum possible value of the sum $\sum_{i \in V} p_i w_i$ where $\{p_1,...,p_n\}$ is a choice of a non-negative real number for each vertex such that $\sum_{i \in C} p_i \leq 1$ for every clique $C \subset V$.

\section{Proofs}

\setcounter{thm}{0}

\begin{thm}
Logical Bell inequalities are derived from CSW inequalities on the exclusivity graph $\cg(\cm, \cc)$.
\end{thm}

\begin{proof}We are given $E_i \subset E(C_i)$ such that the formulae $f_i$ are logically inconsistent.  The corresponding logical Bell inequality is:
$$ \sum_{i = 1,...,N} \sum_{j} p_{ij} \leq N - 1$$
where $N$ is the number of contexts and $p_{ij}$ is the probability of the $j$-th event in $E_i$.
Consider the weighted graph $(G,w)$ with weights 1 for events in $E_i$ and 0 for all others.  The corresponding CSW inequality is:
$$ \sum_{i = 1,...,N} \sum_{j} p_{ij} \leq \alpha(G,w) \, .$$
The quantity $\alpha(G,w)$ is simply the size of the largest independent set of events in $\cup E_i$.  This integer must be strictly less than $N$ for else these events together yield an assignment of \emph{true} or \emph{false} to all the variables in $\cm$ which makes all the formulae $f_i$ true, contradicting their assumed logical inconsistency.\end{proof}

\begin{thm}
\label{hardy}
An empirical model $\empmod$ is logically nonlocal or contextual (i.e. admits an inequality-free proof of nonlocality or contextuality) if and only if the minimal independence number of its support graph $\cg_\empmod$ is less than the number of contexts.
\end{thm}

Suppose an empirical model $\empmod$ gives rise to a support graph $\cg_\empmod$ whose minimal independence number is less than $|\cc|$.  This means that there is a possible event $e: C \to \co$ with independence degree less than $|\cc|$, the number of contexts.  If $\lambda: \cm \to \co$ is a hidden variable compatible with $e$ (that is, $\lambda|_C = e$), then the set of events $I_\lambda = \{\lambda|_C: C \to \co \;|\; C \in \cc\}$ is an independent subset of $\cg(\cm,\cc)$ of size $|\cc|$.  One of these events cannot be in the support graph $\cg_\empmod$ and is therefore impossible.  As there is a possible event $e$ such that every compatible hidden variable must predict an impossible event, $\empmod$ is logically contextual.

Conversely, suppose an empirical model $\empmod$ gives rise to a support graph $\cg_\empmod$ whose minimal independence number is $|\cc|$ and $e$ is a possible event.  The event $e$ must be contained in an independent subset $I$ of $\cg_\empmod$ of size $|\cc|$.  Construct the hidden variable $\lambda_I: \cm \to \co$ by $\lambda(M) = i(M)$ for some $i \in I$ whose context contains $M$.  Every coarse-graining of $\lambda_I$ is in $I \subset \cg_\empmod$ and therefore possible.  So, $\lambda_I$ is a canonical hidden variable that is compatible with $e$ that does not predict an impossible event.  We conclude that  $\empmod$ is not logically contextual.

\begin{thm}
An empirical model $\empmod$ is strongly nonlocal or contextual if and only if the independence number of its support graph $\cg_\empmod$ is less than the number of contexts.
\end{thm}

If the support graph $\cg_\empmod$ of an empirical model $\empmod$ contains an independent set $I$ of size $|\cc|$, then, as above, construct the canonical hidden variable $\lambda_I$.  Each coarse-graining of $\lambda_I$ is in  $I \subset \cg_\empmod$ and thus, is possible.  Therefore, $\empmod$ not is strongly nonlocal or contextual.

Conversely, suppose every independent subset of the support graph $\cg_\empmod$ has size less than $|\cc|$.  If $\lambda: \cm \to \co$ is any canonical hidden variable, then the set $I_\lambda$ of its coarse-grainings $\lambda|_C$ form an independent subset of $\cg(\cm,\cc)$ of size $|\cc|$ and thus at least one $i \in I$ is impossible.  So, $\empmod$ is strongly nonlocal or contextual.

\setcounter{result}{3}
\begin{result}
All $n$-qubit states are strongly contextual with respect to $n$-qubit stabilizer operations whenever $n > 1$.
\end{result}

Consider the measurement scenario consisting of those $n$-qubit Pauli operators whose square is the identity as measurements and, as contexts, those subsets which generate maximal abelian subgroups of the $n$-qubit Pauli group $P_n$.  The support graph of the maximally mixed state, i.e. those formal events which are possible within quantum theory, is the orthogonality graph of $n$-qubit stabilizer quantum mechanics.  It has rank-1 stabilizer projections as vertices with edges joining orthogonal projections.  We must show that the independence number of this graph is less than the number of contexts $|C_n|$.

Suppose for contradiction that $I$ is an independent vertex subset of size $|C_n|$.  $I$ must contain precisely one projector from each context for otherwise it will contain projectors onto two distinct eigenstates for the Paulis of the context which must be orthogonal and thus adjacent.  We denote the projector chosen from a context $C$ by $I_C$.

Thus, $I$ yields a noncontextual value assignment $v: P_n \to \{\pm 1, \pm i\}$.    For an $n$-Pauli $P \in P_n$ in a context $C \in C_n$, $v(P)$ is eigenvalue associated by $P$ to the state $I_C$.  This eigenvalue is independent of the choice of context for if $P$ is a member of both $C$ and $C'$, then it associates the same eigenvalue to both $I_C$ and $I_C'$; otherwise, $I_C$ and $I_C'$ belong to different eigenspaces of $P$ which contradicts the assumption that no pair from $I$ is orthogonal.  We can extend $v$ to all Paulis by multiplicativity.  However, there exists no function $v: P_n \to \{\pm 1, \pm i\}$ mapping Paulis to one of their eigenvalues such that $v(-P) = -v(P)$ and $v(PP') = v(P)v(P')$ whenever $P$ and $P'$ commute.

We first reproduce the standard Mermin argument for 2-qubit systems.  Consider the table of Paulis:

\begin{table}[h]
\centering
\begin{tabular}{|c|c|c|}
\hline
$Z \otimes I$ & $I \otimes X$ & $Z \otimes X$ \\ \hline
$I \otimes Z$ & $X \otimes I$ & $X \otimes Z$ \\ \hline
$Z \otimes Z$ & $X \otimes X$ & $Y \otimes Y$ \\ \hline
\end{tabular}
\end{table}

Each entry has eigenvalues in the set $\{\pm 1\}$.  The rows and columns give commuting triples of $P_2$.  Suppose a valuation $v$ satisying our hypotheses exists and consider the product of $v(P)$ over all $P$ in each row or column.  Since each entry appears in one row and one column, this is the product of $(\pm 1)^2$ nine times: 1.  However, if we collect the terms into rows and columns, we find that we are taking the product of $v(I \otimes I)$ five times with $v(-I \otimes I)$ once.  Therefore, the product is $-1$ and we reach a contradiction.

For the general $n > 2$ case, we repeat the same argument only we tensor all of the entries of the table with $I^{\otimes (n-2)}$ to get entries in $P_n$.  This does not change the eigenvalues of the entries as $I$ has only eigenvalue 1 nor does it affect the commutativity of the rows and columns.  The product of the rows and columns yields $I^{\otimes n}$ five times and $-I^{\otimes n}$ once.

\section{Relation to hypergraph approach}

Here, we show that the graph of Definition 1 coincides with the complement of the non-orthogonality graph [36, 2.3.1]
of the hypergraph [36, D.1.4]
associated by Ac\'in \emph{et al.} to a measurement scenario.

We first recall their relevant definitions before sketching a proof of the result.
\setcounter{defn}{3}
\begin{defn}[D.1.2]
For a measurement $A \in \cm$ in a measurement scenario $(\cm, \cc)$, the \emph{induced measurement scenario} $\cm\{A\}$ has as measurements all $B \in \cm$ such that $A \neq B$ and $\{A,B\}$ is contained within some context $C \in \cc$.  The contexts are the sets $C \setminus \{A\}$ for all $C \in \cc$ such that $A \in C$.
\end{defn}

\begin{defn}[D.1.3]
A \emph{measurement protocol} $T \in MP(T)$ on a measurement scenario $(\cm, \cc)$ is defined recursively as:
\begin{itemize}
\item Base case: $\cm = \emptyset$ \\ $T = \emptyset$
\item Recursive case: $\cm \neq \emptyset$ \\ $T = (A, f)$, where $A \in \cm$ is a measurement and $f: \co \to MP(\cm\{A\})$ is a function from the outcome set to the set of all measurement protocols on $\cm\{A\}$.
\end{itemize}
\end{defn}

\begin{defn}
A \emph{protocol outcome} $\alpha \in POut(T)$ of a measurement protocol $T$ is defined recursively as:
\begin{itemize}
\item Base case: $T = \emptyset$ \\ $\alpha = \{\emptyset\}$
\item Recursive case: $T = (A,f)$ \\ $\alpha = (A,a,\alpha')$, where $a \in \co$ and $\alpha' \in POut(f(a))$ is a protocol outcome of the measurement protocol $f(a)$.
\end{itemize}
A protocol outcome $\alpha$ yields an observable event $s(\alpha): C(\alpha) \to \co$ where $C(\alpha) \in \cc$, also defined recursively, that maps $A$ to $a$ at each stage. 
\end{defn}

\begin{defn}[D.1.4]
The \emph{contextuality scenario} $H[\cm]$ associated to a measurement scenario $(\cm, \cc)$ is a hypergraph with all observable events (functions $e: C \to \co$ from some context $C \in \cc$ to outcomes) as vertices.  Every measurement protocol $T$ on $(\cm, \cc)$ defines a hyperedge in the hypergraph given by $\{s(\alpha) :  \alpha \in POut(T)\}$.
\end{defn}

\begin{defn}[2.3.1]
The \emph{non-orthogonality graph} of a contextuality scenario $H[\cm]$ has the same vertices.  Two vertices share an edge if and only if they are not contained within a common hyperedge.
\end{defn}

\setcounter{lemma}{1}
\begin{lemma}
Let $(\cm, \cc)$ be a measurement scenario.  For any $A \in \cm$ and observable event $e: C \to \co$ where $A \in C \in \cc$, there exists a measurement protocol $T \in MP(\cm)$ such that $T = (A,f)$ and $e = s(\alpha)$ for some $\alpha \in POut(T)$.
\end{lemma}

\begin{proof}
This follows by applying the recursive hypothesis to $\cm\{A\}$, $A' \in (C \setminus A)$, and $e|_{(C \setminus A)} \to \co$ and choosing $f$ to map $e(A)$ to the resulting induced measurement protocol.
\end{proof}

\begin{thm}The exclusivity graph $\cg(\cm, \cc)$ of a measurement scenario $(\cm, \cc)$ is isomorphic to the graph complement of the non-orthogonality graph of $H[\cm]$.
\end{thm}

\begin{proof}
Suppose $e: C \to \co$ and  $e': C' \to \co$ are mutually exclusive events, i.e. there is a measurement $A \in C \cap C'$ such that $e(A) \neq e'(A)$.  We must show that they are contained within a common hyperedge by constructing a measurement protocol $T = (A,f)$ for which $e$ and $e'$ arise as $s(\alpha)$ and $s(\alpha')$ for $\alpha, \alpha' \in POut(T)$.  This can be done by choosing $f$ to map $e(A)$ to the measurement protocols yielded by applying Lemma 1 to $M\{A\}$, $B \in (C \setminus A)$ and $e|_{(C \setminus A)}: (C \setminus A) \to \co$ (and similarly for $e'(A)$).

Conversely, suppose two distinct observable events $e$ and $e'$ are contained within a common hyperedge.  There is thus a measurement protocol $T$ for which $e$ and $e'$ arise as $s(\alpha)$ and $s(\alpha')$ for $\alpha,\alpha' \in POut(T)$.  The minimal depth at which these protocol outcomes diverge gives the common measurement for which $e$ and $e'$ assign different outcomes.
\end{proof}

\end{appendices}

\end{document}